\documentclass[lettersize,journal]{IEEEtran}
\usepackage{amsmath,amsfonts,amssymb}

\DeclareMathOperator{\diag}{diag}
\DeclareMathOperator{\ReOp}{Re}
\DeclareMathOperator{\ImOp}{Im}
\newcommand{\T}{\mathsf{T}}
\usepackage{algorithmic}
\usepackage{array}
\usepackage[caption=false,font=normalsize,labelfont=sf,textfont=sf]{subfig}
\usepackage{textcomp}
\usepackage{stfloats}
\usepackage{url}
\usepackage{verbatim}
\usepackage{graphicx}
\usepackage{bm}  
\usepackage{amsthm}   
\usepackage{cite}     

\usepackage{multirow} 
\usepackage{booktabs} 

\graphicspath{{figure/}}
\newtheorem{proposition}{Proposition}

\def\BibTeX{{\rm B\kern-.05em{\sc i\kern-.025em b}\kern-.08em
    T\kern-.1667em\lower.7ex\hbox{E}\kern-.125emX}}
\usepackage{balance}
\begin{document}
\title{Localization and Reshaping of Non-Minimum-Phase Zeros in Multi-Converter Systems}
\author{Ailixier~Yaermaimaiti,
        Jiaxin~Wang,~\IEEEmembership{Graduate~Student~Member,~IEEE},
        Yunjie~Gu,~\IEEEmembership{Senior~Member,~IEEE},
        and~Huanhai~Xin,~\IEEEmembership{Senior~Member,~IEEE},%
\thanks{This work was supported by the National Natural Science Foundation of China (U24B6008 and U22B6008).}%
\thanks{Ailixier Yaermaimaiti and Huanhai Xin are with the College of Electrical Engineering, Zhejiang University, Hangzhou 310027, China (e-mail: alxr@zju.edu.cn; xinhh@zju.edu.cn).}%
\thanks{Jiaxin~Wang is with the State Key Laboratory of Power System Operation and Control, Department of Electrical Engineering, Tsinghua University, Beijing 100084, China (e-mail: jiaxinwangthu@gmail.com).}%
\thanks{Yunjie Gu is with the Department of Electrical and Electronic Engineering, Imperial College London, London SW7 2AZ, U.K. (e-mail: yunjie.gu@imperial.ac.uk).}}

\maketitle

\begin{abstract}
Non-minimum-phase (NMP) zeros in multi-converter power systems impose bandwidth ceilings on feedback control, yet quantifying them at the system level has been impractical because commercial converters withhold their internal controller models. This paper develops a Jacobian-based framework that decouples the NMP zeros from individual converter dynamics, proves them to be strictly real, and expresses their values as the singular values of a matrix constructed solely from the grid admittance matrix and steady-state power injections. Because these zeros govern the peak magnitude of the complementary sensitivity function, an exponential lower bound on this peak is derived as a function of the dominant zero, establishing that as the zero approaches the origin the stability margin degrades unavoidably. To counteract this degradation, a zero reshaping strategy is proposed that ranks converter nodes by their real participation factors and identifies the optimal site for voltage droop deployment without iterative search, steering the dominant zero away from the origin and thereby suppressing the sensitivity peak.
\end{abstract}

\begin{IEEEkeywords}
Inverter-based resources, non-minimum-phase zero, stability, complementary sensitivity function.
\end{IEEEkeywords}

\section{Introduction}
\label{sec_intro}

\IEEEPARstart{V}{oltage} source converters play a vital role in modern power systems by connecting renewable energy sources, energy storage systems, and high-voltage direct current transmission networks to the grid \cite{Blaabjerg2004,lhlzac2018,Pcatalan2023,ZSui2025}. Their growing penetration, however, weakens the grid by reducing short-circuit ratios and introduces new small-signal stability problems \cite{Hatziargyriou2021,JZzhou2014}. A well-known consequence is the appearance of non-minimum-phase (NMP) zeros in the open-loop transfer functions of grid-tied converters \cite{Zhang2010,FanMiao2018}. These zeros reflect a physical coupling between the converter control inputs and the grid voltage, and they impose hard constraints on the achievable control bandwidth that no feedback design can remove \cite{Skogestad2007,JF1985}. In weak grids, where these zeros shift toward the origin, the resulting bandwidth ceiling becomes the dominant constraint on converter control tuning.

Bode integral theory explains why NMP zeros enforce this bandwidth ceiling \cite{Jiechen2000,Nwan2020}. An NMP zero creates a tradeoff between the control bandwidth and the peak magnitude of the complementary sensitivity function. Pushing the bandwidth toward the magnitude of the NMP zero forces this peak to grow exponentially. Because a large peak reduces the system's tolerance to modeling uncertainties \cite{JDoyle1981,NLehtomaki1981,IPostlethwaite1981}, the bandwidth must be restricted to preserve an adequate stability margin. In practice, the control bandwidth of the converter must remain below approximately half the magnitude of the NMP zero closest to the origin \cite{Wu2021,JinDai2025}.

In single-converter systems, the causes of NMP zeros are well understood. One primary cause is the coupling between phase-angle perturbations and $d$-axis current dynamics. When a converter injects active power into a weak grid, this coupling produces an NMP zero in the open-loop transfer function, and the zero moves toward the origin as the short-circuit ratio decreases \cite{Zhang2010,FanMiao2018,ZhangL2011}. Another major cause involves the $q$-axis current control loop. Here, the grid impedance couples the $q$-axis current to the bus voltage, producing an additional NMP zero that becomes prominent during reactive power exchange under low short-circuit ratio conditions \cite{Yin2023}. In both cases, the resulting bandwidth ceiling severely constrains the dynamic performance and stability margin of the converter \cite{Wu2021,JinDai2025}. However, these single-converter analyses isolate the device from broader network interactions. When multiple converters share the same grid, the NMP zeros are no longer determined by an individual device. Instead, they emerge directly from the physical coupling across the multi-converter system.

In multi-converter systems, quantifying and reshaping NMP zeros presents three major challenges. First, the NMP zeros exhibit a global dependence on the network admittance matrix, nodal voltage phasors, and power injections. A change in the operating point of one converter can shift a zero at a distant bus because the network couples all nodes together \cite{Huang2024,Wdong2019}. This global nature stands in contrast to the single-converter case where zeros are attributed to specific local control loops. Second, commercial converters are typically proprietary. System operators only have access to terminal impedance or frequency-response data and lack the internal state-space models required for a system-level zero analysis \cite{Jfang2025}. Without these detailed models, computing the NMP zeros of the overall system becomes highly restricted. Finally, there is an absence of a systematic method to reshape NMP zero trajectories. Because pushing the critical NMP zero further away from the origin directly suppresses the exponential growth of the complementary sensitivity peak, such a rightward shift is essential to restore the multivariable stability margin. However, even when these zeros are identified, no clear criterion exists to determine which converter requires modification or what local control action maximizes the rightward shift of the critical zero.

To overcome these challenges, this paper develops a Jacobian-based framework for quantifying and reshaping NMP zeros in multi-converter systems. We model the system as a multiple-input-multiple-output (MIMO) feedback interconnection of a network Jacobian plant and device-side controllers, and use the peak magnitude of the complementary sensitivity function as the stability margin metric. From Bode integral theory, we analytically characterize how the dominant NMP zero dictates this peak. We then compute the NMP zeros in closed form through a similarity transformation and Schur complement decomposition of the network Jacobian, reducing the problem to a singular value computation on a matrix built from network admittance data and steady-state power injections. To reshape these zeros, we introduce voltage droop control at selected converter nodes and derive a sensitivity formula that maps the most effective placement to a ranking of nodal participation factors. The main contributions are summarized below.

1) We derive a lower bound on the peak magnitude of the complementary sensitivity function as a function of the dominant zero, revealing a fundamental constraint that explicitly links zero proximity to the origin with the degree of stability margin degradation.

2) We prove that the NMP zeros of a multi-converter system originate exclusively from the network Jacobian and are independent of individual converter controller parameters. This result bypasses the black-box barrier of commercial converters and yields the expression for the NMP zero locations in terms of the singular values of a matrix constructed from the grid admittance and nodal power injections alone.

3) We propose an NMP zero reshaping strategy that identifies the most effective converter node for control deployment by proving that the zero sensitivity to a local voltage droop gain is proportional to the real participation factor at each node, thereby reducing the multi-node control placement problem to a scalar ranking with no iterative search required.

The rest of this paper is organized as follows. Section~\ref{sec:modeling} establishes the multi-converter system model based on the network and device Jacobians. Sections~\ref{sec:sensitivity_constraints} and~\ref{sec:stability_analysis} derive the stability bounds imposed by NMP zeros and quantify these zeros using network data. Section~\ref{sec:reshaping} introduces the NMP zero reshaping strategy for stability enhancement. Simulation results and conclusions follow in Sections~\ref{sec:case_study} and~\ref{sec:conclusion}.

\section{Jacobian-Based Modeling of Multi-Converter Systems}
\label{sec:modeling}

Consider a multi-converter system where $N$ converters are connected through a power network as shown in Fig.~\ref{fig:dj}. The system consists of grid-following (GFL) converters that use Phase-Locked Loops (PLLs) for synchronization. As illustrated in Fig.~\ref{fig:kz}, the control structure consists of outer power loops, inner current loops, and the PLLs.

\begin{figure}[!t]
   \centering
   \includegraphics[width=3.0in]{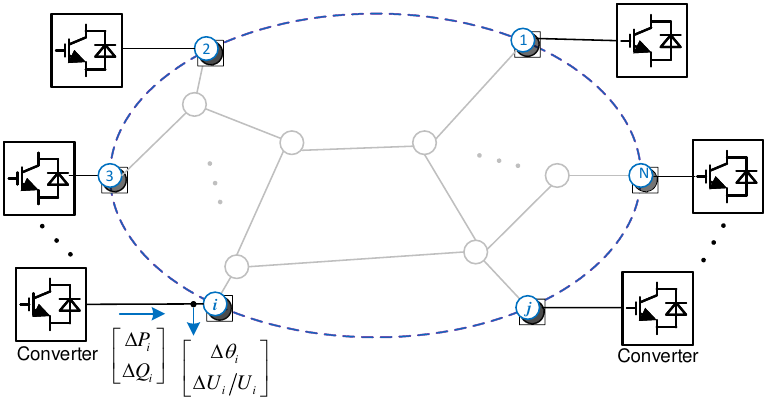}
    \caption{Illustration of a multi-converter system.}
    \label{fig:dj}
\end{figure}

\begin{figure}[!t]
   \centering
   \includegraphics[width=3.0in]{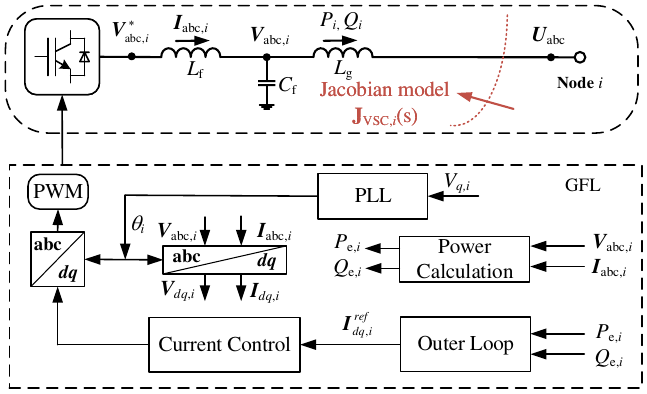}
    \caption{Control structure of a single grid-connected GFL converter.}
    \label{fig:kz}
\end{figure}

For the $i$-th converter, the small-signal dynamics are represented by a Jacobian transfer matrix $\mathbf{J}_{\mathrm{VSC},i}(s)$. This matrix characterizes the linearized relationship between the power perturbations $[\Delta P_i, \Delta Q_i]^\T$ and the voltage coordinates $[\Delta \theta_i, \Delta U_i/U_{i}]^\T$ as
\begin{equation}
\label{eq:single_device}
\begin{aligned}[b]
\begin{bmatrix}
\Delta P_i \\
\Delta Q_i
\end{bmatrix}
&=
\mathbf{J}_{\mathrm{VSC},i}(s)
\begin{bmatrix}
\Delta \theta_i \\
\Delta U_i/U_{i}
\end{bmatrix}, \\[2ex]
\mathbf{J}_{\mathrm{VSC},i}(s) &=
\begin{bmatrix}
J_{P\theta,i}(s) & J_{PU,i}(s) \\
J_{Q\theta,i}(s) & J_{QU,i}(s)
\end{bmatrix}.
\end{aligned}
\end{equation}

The Jacobian transfer matrix $\mathbf{J}_{\mathrm{VSC},i}(s)$ is obtained through analytical modeling or frequency-scanning measurements. We refer to \cite{ZhangL2011} for derivations of $\mathbf{J}_{\mathrm{VSC},i}(s)$. This formulation can also be derived from the impedance/admittance model \cite{Zyang2020}.

We extend \eqref{eq:single_device} to all $N$ converters. A device-side Jacobian matrix $\mathbf{J}_{\mathrm{VSC}}(s)$ is constructed using the Kronecker product to incorporate the capacity scaling matrix $\mathbf{S}_B = \diag\{S_{B1}, \dots, S_{BN}\}$. The aggregated dynamics are expressed as
\begin{equation}
\label{eq:device_agg}
\begin{gathered}
\begin{bmatrix}
\Delta \bm{P} \\
\Delta \bm{Q}
\end{bmatrix}
= \mathbf{J}_{\mathrm{VSC}}(s)
\begin{bmatrix}
\Delta \mathit{\Theta} \\
 \bm{U}^{-1}\Delta \bm{U}
\end{bmatrix}, \\[2ex]
\mathbf{J}_{\mathrm{VSC}}(s) = (\mathbf{I}_2 \otimes \mathbf{S}_B)
\begin{bmatrix}
\mathbf{J}_{P\theta}(s) & \mathbf{J}_{PU}(s) \\
\mathbf{J}_{Q\theta}(s) & \mathbf{J}_{QU}(s)
\end{bmatrix}
\end{gathered}
\end{equation}
where $\bm{P}=[P_1, \dots, P_N]^\T$ and $\bm{Q}=[Q_1, \dots, Q_N]^\T$ are the active and reactive power vectors, respectively; $\mathit{\Theta}=[\theta_1, \dots, \theta_N]^\T$ and $\bm{U}=[U_1, \dots, U_N]^\T$ denote the phase angle and voltage magnitude vectors of the $N$ converters, respectively. $\mathbf{J}_{P\theta}(s)=\diag\{J_{P\theta,1}(s),\dots,J_{P\theta,N}(s)\}$ is the diagonal matrix formed by the corresponding Jacobian elements of each converter. Similar diagonal matrices are defined for $\mathbf{J}_{PU}(s)$, $\mathbf{J}_{Q\theta}(s)$, and $\mathbf{J}_{QU}(s)$.

The network-side Jacobian transfer matrix for the multi-converter system is derived based on a reduced network model. The aggregated network dynamics are formulated as
\begin{equation}
\label{eq:net_jac}
\begin{gathered}
\begin{bmatrix}
\Delta \bm{P} \\
\Delta \bm{Q}
\end{bmatrix}
= \mathbf{J}_{\mathrm{NET}}(s)
\begin{bmatrix}
\Delta \mathit{\Theta} \\
 \bm{U}^{-1}\Delta \bm{U}
\end{bmatrix}, \\[2ex]
\begin{aligned}
\mathbf{J}_{\mathrm{NET}}(s) &=
\begin{bmatrix}
  \alpha(s) & \beta(s) \\
 -\beta(s) & \alpha(s)
\end{bmatrix} \otimes \ReOp(\mathbf{Y}) \\
&\quad +
\begin{bmatrix}
  \beta(s) & -\alpha(s) \\
  \alpha(s) & \beta(s)
\end{bmatrix} \otimes \ImOp(\mathbf{Y})
+ \begin{bmatrix}
 -\mathbf{Q} & \mathbf{P} \\ \mathbf{P} & \mathbf{Q}
\end{bmatrix}.
\end{aligned}
\end{gathered}
\end{equation}

To simplify the network topology, interior nodes not directly connected to the converters are eliminated via Kron reduction under the assumption that the currents injected into these intermediate nodes remain constant \cite{Fdorfler}. This reduction yields the reduced Laplacian matrix $\mathbf{B}_r$, where the off-diagonal $ij$-th element encapsulates the equivalent susceptance between converter node $i$ and node $j$. For derivations regarding the construction of $\mathbf{B}_r$, we refer to \cite{Wdong2019}.

Based on this reduced network, the grid complex susceptance matrix is defined as $\mathbf{Y} = \mathbf{U}\mathbf{B}_r \overline{\mathbf{U}}$, where the overbar indicates the element-wise complex conjugate, while the operators $\ReOp(\cdot)$ and $\ImOp(\cdot)$ extract its real and imaginary parts, respectively. $\mathbf{U}=\diag\left \{ U_1e^{j\theta_1},\dots , U_Ne^{j\theta_N} \right \}$ represents the diagonal matrix formed by the complex voltage phasors of the converter nodes. $\mathbf{P} = \diag\{P_1, \dots, P_N\}$ and $\mathbf{Q} = \diag\{Q_1, \dots, Q_N\}$ are defined as diagonal matrices containing the active and reactive powers of the respective converters. In addition, $\alpha(s)=\frac{1}{(\omega_0/s)^2+1}$ and $\beta(s)=\frac{\omega_0/s}{(\omega_0/s)^2+1}$, where $\omega_0$ is the nominal angular frequency of the system. The symbol $\otimes$ denotes the Kronecker product. The mathematical derivation of the network Jacobian matrix is provided in Appendix~\ref{app:jacobian}.

To establish the closed-loop feedback interconnection depicted in Fig.~\ref{fig:fd}, the device cluster is modeled as the forward path controller. This requires inverting the aggregated relationship in \eqref{eq:device_agg} to express voltage perturbations as a function of power variations, so the controller $\mathbf{K}_{\mathrm{VSC}}(s)$ is defined as the inverse of the device-side Jacobian matrix
\begin{equation}
\label{eq:K_inv}
\mathbf{K}_{\mathrm{VSC}}(s) = \mathbf{J}_{\mathrm{VSC}}^{-1}(s).
\end{equation}

Combining the network model $\mathbf{J}_{\mathrm{NET}}(s)$ with the inverted device model $\mathbf{K}_{\mathrm{VSC}}(s)$ yields the equivalent feedback block diagram shown in Fig.~\ref{fig:fd}.

\begin{figure}[!t]
   \centering
    \includegraphics[width=3.0in]{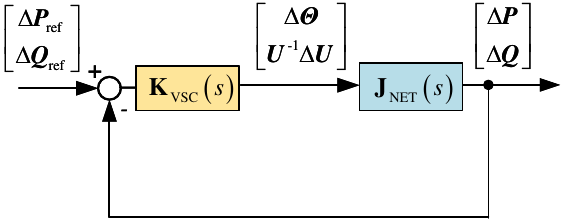}
    \caption{Equivalent feedback block diagram of the Jacobian model.}
    \label{fig:fd}
\end{figure}

\section{Complementary Sensitivity Function Constraints Imposed by NMP Zeros}
\label{sec:sensitivity_constraints}

\subsection{Bode Integral Constraint on the Complementary Sensitivity Function}

As shown in Fig.~\ref{fig:fd}, the open-loop transfer function of the multi-converter system is $\mathbf{L}(s) = \mathbf{J}_{\mathrm{NET}}(s)\mathbf{K}_{\mathrm{VSC}}(s)$. A measure of the system's stability margin in the presence of modeling errors is the complementary sensitivity matrix $\mathbf{T}(s)$, expressed as
\begin{equation}\label{eq:complementary_sensitivity}
    \mathbf{T}(s) = \mathbf{L}(s)(\mathbf{I} + \mathbf{L}(s))^{-1}.
\end{equation}

According to the Small Gain Theorem \cite{Skogestad2007}, a smaller peak magnitude of the maximum singular value $\bar{\sigma}(\mathbf{T}(j\omega))$ yields a larger tolerance to uncertainties, serving as a robust measure of stability. Conversely, a lower tolerance to uncertainties degrades the stability margin and renders the system susceptible to instability.

However, shaping $\mathbf{T}(s)$ is constrained by the plant's NMP characteristics. According to Bode integral theory \cite{Jiechen2000}, NMP zeros impose an integral constraint on the complementary sensitivity function
\begin{equation} \label{eq:integral_constraint}
    \int_{0}^{\infty} \ln \bar{\sigma}(\mathbf{T}(j\omega)\mathbf{T}^{-1}(0)) \frac{d\omega}{\omega^2} \ge \frac{\pi}{2} \bar{\lambda} \left( \sum_{i=1}^{l} \frac{2\ReOp(z_i)}{|z_i|^2} \mathbf{w}_i \mathbf{w}_i^{\mathrm{H}} + \mathbf{C} \right)
\end{equation}
where $z_i$ denotes the $i$-th NMP zero of the system, defined as the roots of $\det(\mathbf{J}_{\mathrm{NET}}(z_i)) = 0$ with $\ReOp(z_i) > 0$. The vector $\mathbf{w}_i$ is the corresponding unit output zero direction satisfying $\mathbf{w}_i^{\mathrm{H}} \mathbf{J}_{\mathrm{NET}}(z_i) = 0$ and $\|\mathbf{w}_i\| = 1$. The constant matrix $\mathbf{C} = \frac{1}{2}(\mathbf{T}'(0)\mathbf{T}^{-1}(0) + (\mathbf{T}'(0)\mathbf{T}^{-1}(0))^{\mathrm{H}})$ reflects the low-frequency tracking performance. Note that this constraint assumes the controller $\mathbf{K}_{\mathrm{VSC}}(s)$ is stable and does not introduce unstable pole-zero cancellations.

From \eqref{eq:integral_constraint}, if an NMP zero $z_i$ has a small magnitude, the lower bound of the integral becomes large. To satisfy this constraint, the maximum singular value $\bar{\sigma}(\mathbf{T}(j\omega))$ increases over intermediate frequency ranges. The peak magnitude of the complementary sensitivity matrix, $M_T = \|\mathbf{T}(s)\|_\infty = \sup_\omega \bar{\sigma}(\mathbf{T}(j\omega))$, therefore serves as an indicator of system stability. A high $M_T$ reduces the multivariable stability margin by limiting the system's tolerance to uncertainties.

To quantify the trade-off between NMP zeros and $M_T$, we simplify \eqref{eq:integral_constraint} by considering the typical low-frequency behavior of the control system. First, practical feedback loops typically maintain a very high loop gain at low frequencies to ensure steady-state tracking, which makes $\mathbf{L}(0)$ sufficiently large and yields $\mathbf{T}(0) \approx \mathbf{I}$. Second, control design generally requires the complementary sensitivity function to be flat and smooth in the low-frequency range, implying that its derivative at $s=0$ approximates a null matrix, i.e., $\mathbf{T}'(0) \approx \mathbf{0}$. Substituting these low-frequency properties into \eqref{eq:integral_constraint} yields an explicit lower bound on the peak $M_T$ as a function of the NMP zeros alone.

\begin{proposition}
\label{prop:MT_bound}
Under the assumptions $\mathbf{T}(0) \approx \mathbf{I}$ and $\mathbf{T}'(0) \approx \mathbf{0}$, the peak magnitude $M_T$ satisfies the exponential lower bound
\begin{equation} \label{eq:MT_bound_mimo}
    M_T \ge \exp \left( \frac{\pi}{4} \omega_c \bar{\lambda} \left( \sum_{i=1}^{l} \frac{2\ReOp(z_i)}{|z_i|^2} \mathbf{w}_i \mathbf{w}_i^{\mathrm{H}} \right) \right)
\end{equation}
where $\omega_c = \arg \max_{\omega \in (0, \infty)} \frac{\ln \bar{\sigma}(\mathbf{T}(j\omega))}{\omega^2}$. Equation \eqref{eq:MT_bound_mimo} shows that NMP zeros with small magnitude drive an exponential increase in $M_T$, so these NMP dynamics impose an irreducible limit on the system's stability margin.
\end{proposition}
\begin{proof}
Under the stated assumptions, the integral in \eqref{eq:integral_constraint} reduces to a positive lower bound determined by the NMP zeros. Let $\mathcal{K}_z = \frac{\pi}{2} \bar{\lambda} \left( \sum_{i=1}^{l} \frac{2\ReOp(z_i)}{|z_i|^2} \mathbf{w}_i \mathbf{w}_i^{\mathrm{H}} \right)$. Partitioning the integration interval at the characteristic frequency $\omega_c$ yields
\begin{equation}
    \int_{0}^{\omega_c} \frac{\ln \bar{\sigma}(\mathbf{T}(j\omega))}{\omega^2} d\omega + \int_{\omega_c}^{\infty} \frac{\ln \bar{\sigma}(\mathbf{T}(j\omega))}{\omega^2} d\omega \ge \mathcal{K}_z.
\end{equation}

To establish the relationship with the peak magnitude $M_T$, we bound the singular value $\bar{\sigma}(\mathbf{T}(j\omega))$ by its supremum $M_T$ across both the low-frequency and high-frequency ranges, yielding
\begin{equation}
\label{eq:approx_bounds}
\begin{aligned}
    \int_{0}^{\omega_c} \frac{\ln \bar{\sigma}(\mathbf{T}(j\omega))}{\omega^2} d\omega &\le \int_{0}^{\omega_c} \frac{\ln M_T}{\omega_c^2} d\omega  \\
    \int_{\omega_c}^{\infty} \frac{\ln \bar{\sigma}(\mathbf{T}(j\omega))}{\omega^2} d\omega &\le \int_{\omega_c}^{\infty} \frac{\ln M_T}{\omega^2} d\omega.
\end{aligned}
\end{equation}

Combining these two upper bounds gives $\frac{2\ln M_T}{\omega_c} \ge \mathcal{K}_z$. Rearranging this inequality directly establishes the theoretical exponential lower bound imposed by the NMP zeros
\begin{equation} \label{eq:MT_final_bound}
    M_T \ge \exp \left( \frac{\omega_c}{2} \mathcal{K}_z \right). \qedhere
\end{equation}
\end{proof}

\subsection{Numerical Analysis With a $2 \times 2$ MIMO Example}

To validate the performance constraints imposed by NMP zeros, we examine a $2 \times 2$ MIMO feedback control system. The plant model $\mathbf{J}(s)$ and controller $\mathbf{K}(s)$ are given by
\begin{equation}
    \mathbf{J}(s) = \begin{bmatrix}  \frac{5\left( 1- s/z\right)}{s+10} & \frac{0.5}{s+20} \\ \frac{0.5}{s+20} & \frac{5}{s+20} \end{bmatrix},
    \mathbf{K}(s) = \left( K_p + \frac{K_i}{s} \right) \frac{\mathbf{I}_2}{0.05 s + 1}
\end{equation}
where $z$ denotes a real NMP zero and $\mathbf{I}_2$ denotes the $2 \times 2$ identity matrix. The controller parameters are set to $K_p=5$ and $K_i=50$.

We first verify the assumptions via the frequency response of the complementary sensitivity function $\mathbf{T}(j\omega)$, as shown in Fig.~\ref{fig:assumptions}. The magnitude curves originate at 0 dB with a flat slope near the origin. This confirms that $\mathbf{T}(0) = \mathbf{I}$ and $\mathbf{T}'(0) \approx \mathbf{0}$, satisfying the stated assumptions.

Subsequently, we investigate the impact of the NMP zero location ($z \in \{40, 60, 80\}$) on stability. Figure~\ref{fig:stability}(a) illustrates that as the zero $z$ moves closer to the origin, the peak magnitude $M_T$ increases. This elevated peak corresponds to a reduced stability margin, manifested as oscillations in the time-domain response for the $z=40$ case as shown in Fig.~\ref{fig:stability}(b).

To expose the limitations of conventional controller tuning, we evaluate an extreme case where the NMP zero approaches the origin ($z=0.01$). As detailed in Table~\ref{tab:extreme_zero}, stabilizing such a system requires a drastic reduction in PI gains, which severely compresses the closed-loop bandwidth to near zero (0.0025 rad/s). This extreme bandwidth penalty deprives the converter of its dynamic tracking capability, demonstrating that no feasible tuning strategy can stabilize the system without sacrificing practical performance.

\begin{table}[!t]
\renewcommand{\arraystretch}{1.3}
\caption{Dominant closed-loop eigenvalue of the $2\times 2$ MIMO system with an NMP zero at $z=0.01$}
\label{tab:extreme_zero}
\centering
\begin{tabular}{cccc}
\toprule
$K_p$ & $K_i$ & Bandwidth (rad/s) &Dominant Mode \\
\midrule
1.000  & 10.000 & 2.5283& $9979.99$ \\
0.100  & 1.000  & 0.2500& $979.98$ \\
0.010  & 0.100  & 0.0249& $79.98$ \\
0.001  & 0.010  & 0.0025& $-0.0025$ \\
\bottomrule
\end{tabular}
\end{table}

Finally, we verify that the conservativeness of the $M_T$ lower bound is practically acceptable. Figure~\ref{fig:bound} compares the actual sensitivity peaks against the theoretical bounds for $z \in \{60, 80\}$. The measured peaks closely track the dashed theoretical constraints, confirming that NMP zeros provide a reliable metric for evaluating performance boundaries.

\begin{figure}[!t]
   \centering
    \includegraphics[width=3.0in]{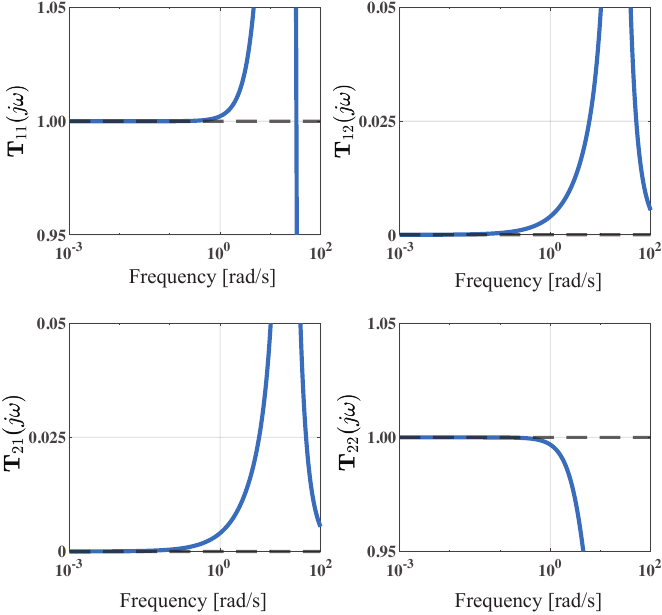}
    \caption{Low-frequency response of $\mathbf{T}(j\omega)$.}
    \label{fig:assumptions}
\end{figure}

\begin{figure}[!t]
   \centering
    \includegraphics[width=3.0in]{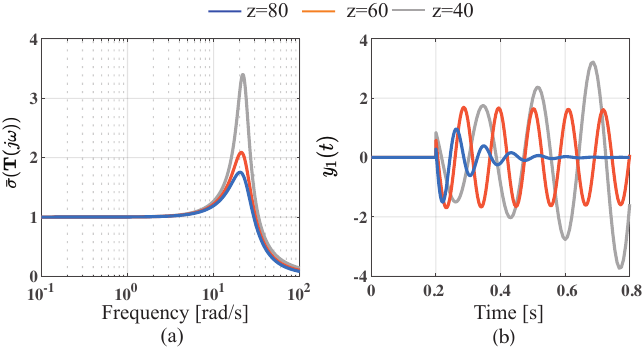}
    \caption{Impact of NMP zero location ($z=40, 60, 80$) on system stability. (a) Singular value plots of $\mathbf{T}(j\omega)$. (b) Time-domain responses of the system.}
    \label{fig:stability}
\end{figure}

\begin{figure}[!t]
   \centering
    \includegraphics[width=3.0in]{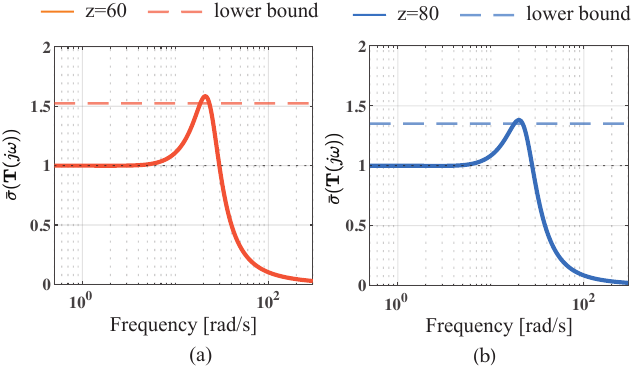}
    \caption{The peak magnitudes of the complementary sensitivity function and the theoretical lower boundaries of the system ($K_p=5$ and $K_i=30$): (a) $z=60$, and (b) $z=80$.}
    \label{fig:bound}
\end{figure}

\section{Quantification of NMP Zeros From Network Data}
\label{sec:stability_analysis}

\subsection{Decoupling of NMP Zeros From Converter Models}

To analyze the stability boundaries, we investigate the NMP zeros of the system. The transmission zeros of the open-loop system are determined by the roots of the determinant equation $\det(\mathbf{L}(s)) = 0$. By substituting the feedback relationship $\mathbf{K}_{\mathrm{VSC}}(s) = \mathbf{J}_{\mathrm{VSC}}^{-1}(s)$, the determinant is expanded as follows
\begin{equation}
\label{eq:det_analysis}
\det(\mathbf{L}(s)) = \det(\mathbf{J}_{\mathrm{NET}}(s)) \cdot \det(\mathbf{J}_{\mathrm{VSC}}^{-1}(s)).
\end{equation}

It can be seen that the zeros of the aggregated controller $\mathbf{K}_{\mathrm{VSC}}(s)$ correspond directly to the poles of the device Jacobian $\mathbf{J}_{\mathrm{VSC}}(s)$ in \eqref{eq:det_analysis}. Under the standard power system assumption that both the passive network plant and the individual device controllers are inherently stable, $\mathbf{J}_{\mathrm{VSC}}(s)$ contains no unstable poles, and its inverse $\mathbf{K}_{\mathrm{VSC}}(s)$ introduces no NMP zeros.

Any NMP zeros in the aggregated multi-converter system therefore originate strictly from the network plant, and the stability assessment reduces to finding the roots of the network Jacobian
\begin{equation}
\label{eq:net_zero_condition}
\det(\mathbf{J}_{\mathrm{NET}}(s)) = 0.
\end{equation}

This decoupling eliminates the need for proprietary control models of commercial VSCs. Because the NMP zeros are dictated entirely by $\mathbf{J}_{\mathrm{NET}}(s)$, their locations can be determined from physical network parameters and steady-state operating points alone.

Operationally, $\det(\mathbf{J}_{\mathrm{NET}}(s))=0$ identifies frequencies of network rank loss. Since its computation relies solely on routine steady-state data, NMP zeros emerge as inherent network properties. This allows operators to monitor the dominant zero in real time using updated operating points, treating converters as black-box injections and eliminating the need for proprietary controller models.

\subsection{NMP Zeros Quantification Using Network Parameters and Operating Points}

To derive an explicit analytical expression for these zeros, it is necessary to decouple the frequency variable $s$ from the network parameters. We achieve this by applying a similarity transformation to $\mathbf{J}_{\mathrm{NET}}(s)$ using the matrix $\mathbf{W}=\frac{1}{\sqrt{2}}\left[\begin{smallmatrix}1 & j \\1 & -j\end{smallmatrix}\right]\otimes \mathbf{I}_N$. This transformation inherently preserves the determinant and does not affect the system zeros, satisfying the strict equivalence condition $\det\!\big(\mathbf{J}_{\mathrm{NET}}(s)\big)=0 \Leftrightarrow \det\!\big(\mathbf{W}\mathbf{J}_{\mathrm{NET}}(s)\mathbf{W}^{-1}\big)=0$. The resulting transformed matrix is directly expressed as
\begin{equation}
\label{eq:transformed_J}
\mathbf{W}\mathbf{J}_{\mathrm{NET}}(s)\mathbf{W}^{-1} =
\begin{bmatrix}
\bar{\gamma}(s)\mathbf{Y} & j\mathbf{S} \\
-j\overline{\mathbf{S}} & \gamma(s)\overline{\mathbf{Y}}
\end{bmatrix} = \begin{bmatrix} \mathbf {A}_{11} & \mathbf {A}_{12} \\ \mathbf {A}_{21} & \mathbf {A}_{22} \end{bmatrix}.
\end{equation}
where $\mathbf{S} = \mathbf{P} + j\mathbf{Q} = \diag\{ S_1e^{j\phi_1}, \dots, S_Ne^{j\phi_N} \}$ represents the complex apparent power matrix, with $S_i$ and $\phi_i$ denoting the magnitude and phase angle of the power injection at node $i$, respectively. The complex scalar function is defined as $\gamma(s)=\alpha(s)+j\beta(s)$.

Assuming the sub-block $\mathbf {A}_{22}$ is invertible, we utilize the Schur complement to evaluate the determinant \cite{Fzhang2005}, i.e., $\det(\mathbf{W}\mathbf{J}_{\mathrm{NET}}(s)\mathbf{W}^{-1}) = \det(\mathbf {A}_{22})\det(\mathbf {A}_{11}-\mathbf {A}_{12}\mathbf {A}_{22}^{-1}\mathbf {A}_{21})$. The condition for identifying a system zero then reduces to finding the roots where the determinant of the Schur complement vanishes
\begin{equation}
\label{eq:schur_condition}
\begin{split}
0 & = \det(\mathbf{A}_{11}-\mathbf{A}_{12}\mathbf{A}_{22}^{-1}\mathbf{A}_{21}) \\
& = \prod_{i=1}^{N} \left( \lambda_i \left( \mathbf{S}^{-1}\mathbf{Y}\overline{\mathbf{S}}^{-1}\overline{\mathbf{Y}} \right) - \frac{1}{\bar{\gamma}(s)\gamma(s)} \right).
\end{split}
\end{equation}

Equation \eqref{eq:schur_condition} demonstrates that the determinant factors into a product of eigenvalues. By defining $\lambda_i(\cdot)$ as the $i$-th eigenvalue operator, finding the system zeros translates to solving $N$ independent scalar equations
\begin{equation}
\label{eq:eigen_problem}
z_i^2 + 1 = \lambda_i \left( \mathbf{S}^{-1}\mathbf{Y}\overline{\mathbf{S}}^{-1}\overline{\mathbf{Y}} \right).
\end{equation}

To analyze the structural properties governing these roots, we decompose the matrix product $\mathbf{S}^{-1}\mathbf{Y}\overline{\mathbf{S}}^{-1}\overline{\mathbf{Y}}$ to isolate the network topology from the steady-state operating points. This decomposition establishes the eigenvalue equivalence $\lambda_i(\mathbf{S}^{-1}\mathbf{Y}\overline{\mathbf{S}}^{-1}\overline{\mathbf{Y}}) = \lambda_i(\mathbf{D}\mathbf{B}_r\overline{\mathbf{D}}\mathbf{B}_r)$, where the diagonal matrix $\mathbf{D} = \diag\left\{\frac{(U_1e^{j\theta_1})^2}{S_1e^{j\phi_1}},\dots,\frac{(U_Ne^{j\theta_N})^2}{S_Ne^{j\phi_N}}\right\}$ captures the operating state.

Substituting this relationship into the eigenvalue expression yields the term $\lambda_i(\mathbf{D}\mathbf{B}_r\overline{\mathbf{D}}\mathbf{B}_r)$. Noting the symmetry $\mathbf B_r=\mathbf B_r^{\mathrm{H}}$ and the conjugate transpose property $\overline{\mathbf{D}}=\mathbf{D}^{\mathrm{H}}$, we apply a similarity transformation using $\mathbf{B}_r^{1/2}$. This transformation symmetrizes the expression and maps the eigenvalues directly to the squares of the corresponding singular values.
\begin{equation}
\label{eq:similarity}
\begin{aligned}
\lambda_i(\mathbf{D}\mathbf{B}_r\overline{\mathbf{D}}\mathbf{B}_r) &= \lambda_i\left( \mathbf{B}_r^{1/2} (\mathbf{D}\mathbf{B}_r\overline{\mathbf{D}}\mathbf{B}_r) \mathbf{B}_r^{-1/2} \right) \\
&= \sigma_i^2\left( \mathbf{B}_r^{1/2}\mathbf{D}\mathbf{B}_r^{1/2} \right).
\end{aligned}
\end{equation}
where $\sigma_i(\cdot)$ specifically denotes the $i$-th singular value operator.

Finally, substituting this singular value mapping back into \eqref{eq:eigen_problem} yields the explicit formulation for the critical NMP zero $z_i$ as
\begin{equation}
\label{eq:zero_final}  
z_i =\omega_0 \sqrt{\sigma_i^2\left( \mathbf{B}_r^{1/2}\mathbf{D}\mathbf{B}_r^{1/2} \right) - 1}.
\end{equation}

Because singular values $\sigma_i$ are real and non-negative, the NMP zeros of the multi-converter system are strictly real. An unstable zero ($z_i > 0$) exists if and only if the structural threshold $\sigma_i\left( \mathbf{B}_r^{1/2}\mathbf{D}\mathbf{B}_r^{1/2}\right) > 1$ is breached, a condition that is readily satisfied in practical multi-converter systems.

\section{NMP Zero Reshaping for Stability Enhancement}
\label{sec:reshaping}

\subsection{Lower Bound on the Complementary Sensitivity Peak Governed by the Dominant NMP Zero}

Because the NMP zeros of the multi-converter system are strictly real, as established in Section~\ref{sec:stability_analysis}, each zero lies on the positive real axis and can be ordered by its distance from the origin. The smallest among them, termed the dominant NMP zero, therefore governs the spectral bound in a scalar sense, reducing the MIMO stability constraint to a single scalar quantity. With all zeros real, the bound in \eqref{eq:MT_bound_mimo} simplifies to
\begin{equation}
\label{eq:MT_bound_real}
    M_T \ge \exp \left( \frac{\pi}{4} \omega_c \bar{\lambda} \left( \sum_{i=1}^{l} \frac{2}{z_i} \mathbf{w}_i \mathbf{w}_i^{\mathrm{H}} \right) \right).
\end{equation}

As \eqref{eq:MT_bound_real} shows, the peak magnitude of the complementary sensitivity function, $M_T$, is constrained by the NMP zeros of $\mathbf{J}_{\mathrm{NET}}(s)$. Under weak grid conditions, a critical real NMP zero $z_{0}$ approaches the origin and this constraint tightens, degrading stability margins. Given that this dominant zero lies significantly closer to the origin than all other zeros ($z_{0} \ll |z_i|$ for $i \neq 0$), its reciprocal term $1/z_0$ overwhelmingly dictates the summation. Because the associated output direction $\mathbf{w}_0$ is a unit vector ($\|\mathbf{w}_0\|_2=1$), the outer product $\mathbf{w}_0 \mathbf{w}_0^{\mathrm{H}}$ forms a rank-1 projection matrix whose maximum eigenvalue is exactly 1. By isolating this dominant term, the spectral radius constraint reduces to
\begin{equation}
\label{eq:MT_bound_dominant}
 \bar{\lambda} \left( \sum_{i=1}^{l} \frac{2}{z_i} \mathbf{w}_i \mathbf{w}_i^{\mathrm{H}} \right) \approx \frac{2}{z_0}.
\end{equation}

Substituting \eqref{eq:MT_bound_dominant} into \eqref{eq:MT_bound_real} yields the simplified lower bound
\begin{equation}
\label{eq:MT_bound_scalar}
    M_T \ge \exp \left( \frac{\pi \omega_c}{2 z_0} \right).
\end{equation}

Equation \eqref{eq:MT_bound_scalar} demonstrates that the lower bound of $M_T$ grows exponentially as the dominant zero $z_0$ approaches the origin. A high $M_T$ produces severe resonant peaks and erodes stability margins, so the control strategy must shift $z_0$ away from the origin to preserve stable operation.

\subsection{NMP Zero Reshaping Strategy via Participation Factors}

To shift the NMP zero $z_0$ further away from the origin, a voltage droop control loop is introduced to the grid-following converters. The linearized local control law for the $i$-th converter is defined as $\Delta Q_i = -k_i \Delta U_i$, where $k_i$ is the droop gain. The implementation of this voltage droop control is adopted from \cite{Lhuang2020}.

Physically, this control measure establishes a direct proportional feedback path from the voltage magnitude to the reactive power injection. As illustrated in Fig.~\ref{figxiachui}, this closed-loop mechanism is mathematically equivalent to paralleling a feedback gain matrix $\mathbf{K} = \diag(0, \dots, k_i, \dots, 0)$ across the open-loop network plant $\mathbf{J}_{\mathrm{NET}}(s)$, yielding the modified system Jacobian
\begin{equation} \label{eq:matrix_perturbation}
    \mathbf{J}_{\mathrm{sys}}(s) = \mathbf{J}_{\mathrm{NET}}(s) + \mathbf{K}.
\end{equation}
\begin{figure}[!t]
\centering
\includegraphics[width=3.0in]{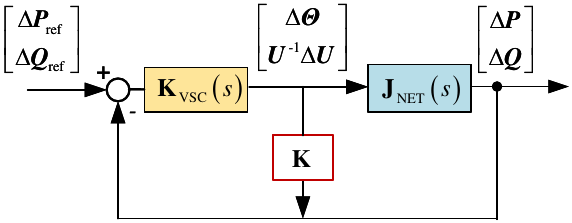}
\caption{Schematic of the Jacobian perturbation where voltage control acts as a feedback gain $\mathbf{K}$ added to the network plant $\mathbf{J}_{\mathrm{NET}}$.}
\label{figxiachui}
\end{figure}

The control placement objective is to identify the specific node $i$ where the control effort $k_i$ yields the maximum rightward shift of the critical zero
\begin{equation}
    i_{\text{opt}} = \arg \max_{i} \ReOp\left( \frac{\mathrm{d}z_0}{\mathrm{d}k_i} \right).
\end{equation}

Applying the Implicit Function Theorem \cite{Magnus2019} to $\det(\mathbf{J}_{\mathrm{sys}}(z_0, k_i)) = 0$ gives the analytical sensitivity of the critical zero with respect to the local gain as
\begin{equation} \label{eq:sensitivity_breakdown}
    \frac{\mathrm{d}z_0}{\mathrm{d}k_i} = - \frac{\mathbf{l}^{\mathrm{H}} \frac{\partial \mathbf{J}_{\mathrm{sys}}}{\partial k_i} \mathbf{r}}{\mathbf{l}^{\mathrm{H}} \frac{\partial \mathbf{J}_{\mathrm{sys}}}{\partial s} \mathbf{r}}.
\end{equation}
where $\mathbf{l}$ and $\mathbf{r}$ denote the left and right eigenvectors of the Jacobian $\mathbf{J}_{\mathrm{sys}}$ evaluated at $k=0$ corresponding to the zero $z_0$.

For a single-node perturbation, the partial derivative in the numerator of \eqref{eq:sensitivity_breakdown} simplifies to $\mathbf{e}_i \mathbf{e}_i^\T$, reducing the entire numerator to the geometric participation factor $p_i = l_i r_i$ associated with the $i$-th node. Meanwhile, the denominator acts as a global scaling factor that remains invariant regardless of the chosen control location $i$.

When converters operate under standard active and reactive power control, the global scaling factor in the denominator of \eqref{eq:sensitivity_breakdown} is strictly positive real for any passive power network, as proved in Appendix~\ref{appendix:proof}. With this positive scaling established, the local sensitivity reduces to
\begin{equation}
    \frac{\mathrm{d}z_0}{\mathrm{d}k_i} \propto \ReOp(p_i).
\end{equation}

The node with the largest real participation factor therefore receives the voltage droop control, maximizing the rightward zero shift per unit of control effort
\begin{equation}
    i_{\text{opt}} = \arg \max_{i} \ReOp(p_i).
\end{equation}

\section{Case Studies}
\label{sec:case_study}

\subsection{Impact of Operating Points on NMP Zeros and System Stability}

To verify the proposed theoretical framework, time-domain simulations and frequency analyses are conducted on a modified IEEE 9-bus system built in MATLAB/Simulink. As illustrated in Fig.~\ref{fig:system_topology}, the testbed comprises three VSCs operating under active and reactive power (PQ) outer-loop control. The system operates at a nominal frequency of 50 Hz. The initial steady-state active and reactive power injections for the three VSCs are set as $(P_1, Q_1) = (0.9, 0.44)$~p.u., $(P_2, Q_2) = (0.8, 0.6)$~p.u., and $(P_3, Q_3) = (1.0, 0)$~p.u., with the corresponding voltage phase angles being $\theta_1 = 0.05$~rad, $\theta_2 = 0.1$~rad, and $\theta_3 = 0.15$~rad, respectively. The parameters of the VSCs are listed in Table~\ref{tab:hardware_parameters} and Table~\ref{tab:control_parameters}, while the parameters of the power network are provided in Appendix~\ref{app:parameters}.

\begin{table}[htbp]
\renewcommand{\arraystretch}{1.3}
\centering
\caption{Main Circuit and Hardware Parameters of VSCs}
\label{tab:hardware_parameters}
\begin{tabular}{llc}
\toprule
\multicolumn{1}{c}{\textbf{Symbol}} & \multicolumn{1}{c}{\textbf{Parameter}} & \textbf{Value} \\
\midrule
$S, U_{\mathrm{B}}$    & Base capacity and voltage of AC system & $1500$\,kVA, $690$\,V \\
$S_{\mathrm{B}}$       & Base capacity of VSC             & $1500$\,kVA \\
$L_{\mathrm{f}}$       & Filter inductance                & $0.05$\,p.u. \\
$C_{\mathrm{f}}$       & Filter capacitance               & $0.05$\,p.u. \\
\bottomrule
\end{tabular}
\end{table}

\begin{table}[htbp]
\renewcommand{\arraystretch}{1.3}
\centering
\caption{Control Parameters of VSCs}
\label{tab:control_parameters}
\begin{tabular}{lc}
\toprule
\multicolumn{1}{c}{\textbf{Parameter}} & \textbf{Value} \\
\midrule
Power control loop of VSCs $H_{\mathrm{PQ}}(s)$ & $0.4 + 5/s$ \\
PLL PI controller of VSC 1 $H_{\mathrm{PLL,1}}(s)$ & $30 + 9000/s$ \\
PLL PI controller of VSC 2 $H_{\mathrm{PLL,2}}(s)$ & $25 + 8100/s$ \\
PLL PI controller of VSC 3 $H_{\mathrm{PLL,3}}(s)$ & $12 + 7000/s$ \\
Voltage feedforward filter $H_{\mathrm{FF}}(s)$ & $1/(1+0.001s)$ \\
\bottomrule
\end{tabular}
\end{table}

\begin{figure}[!t]
\centering
\includegraphics[width=3.0in]{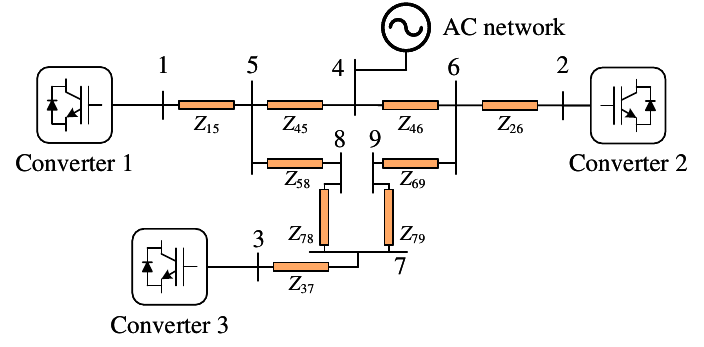}
\caption{Diagram of the modified 3-machine 9-bus system with three PQ-controlled VSCs.}
\label{fig:system_topology}
\end{figure}

The active power output of VSC3 is progressively increased across three operating scenarios, namely Case~1 ($P_3 = 0.8$~p.u.), Case~2 ($P_3 = 0.9$~p.u.), and Case~3 ($P_3 = 1.0$~p.u.). The critical NMP zeros are computed analytically from \eqref{eq:zero_final}, where $\mathbf{D} = \diag\{(U_i e^{j\theta_i})^2 / (S_i e^{j\phi_i})\}$ encodes the operating state and $\mathbf{B}_r$ is the reduced nodal susceptance matrix. For Case~1, these matrices evaluate to $ \mathbf{D}=\diag\{1.11 \angle{-9.56^\circ},\; 1.11 \angle{-23.08^\circ},\; 0.74 \angle{73.36^\circ}\}$,
$$ \mathbf{B}_r = \begin{bmatrix} 6.91 & -0.43 & -0.16 \\ -0.43 & 8.97 & -0.46 \\ -0.16 & -0.46 & 2.03 \end{bmatrix}$$
from which the dominant singular value yields the critical NMP zero $z_0=640$~rad/s. Table~\ref{tab:op_conditions} summarizes the steady-state voltage profiles and the analytically computed NMP zeros under all three cases.

\begin{table}[!t]
\renewcommand{\arraystretch}{1.3}
\caption{Steady-State Profiles and NMP Zeros Under Different Operating Conditions}
\label{tab:op_conditions}
\centering
\begin{tabular}{lcccc}
\toprule
Condition & $U_1$ (p.u.) & $U_2$ (p.u.) & $U_3$ (p.u.) & NMP Zero \\
\midrule
Case~1 & $1.06 \angle 8.04^\circ$ & $1.06 \angle 6.51^\circ$ & $0.95 \angle 26.16^\circ$ & 640 \\
Case~2 & $1.06 \angle 8.13^\circ$ & $1.06 \angle 6.69^\circ$ & $0.92 \angle 30.64^\circ$ & 495 \\
Case~3 & $1.06 \angle 8.25^\circ$ & $1.05 \angle 6.89^\circ$ & $0.86 \angle 36.68^\circ$ & 341 \\
\bottomrule
\end{tabular}
\end{table}

As Table~\ref{tab:op_conditions} shows, the critical NMP zero drops monotonically from 640 to 341~rad/s as $P_3$ increases from 0.8 to 1.0~p.u. This leftward zero migration directly constrains the achievable stability margin through the exponential bound $M_T \ge \exp(\pi \omega_c / 2z_0)$ established in \eqref{eq:MT_bound_scalar}. As Fig.~\ref{fig:msv_peak} confirms, the peak magnitude of the maximum singular value of the complementary sensitivity function increases from 26 dB in Case~1 to 43 dB in Case~2. The unstable scenario (Case~3, $P_3 = 1.0$~p.u.) is omitted from this comparison because the sensitivity peak loses its physical validity for an internally unstable closed-loop system. This escalation correlates with the leftward migration of the critical NMP zero from $z_0 = 640$ to $z_0 = 495$~rad/s. According to the Small Gain Theorem, a larger peak magnitude dictates a lower tolerance to multiplicative uncertainties. Consequently, the observed amplification in the sensitivity peak indicates a qualitative degradation of the robust stability margin, rendering the system increasingly susceptible to instability as the operating condition becomes more stressed.

\begin{figure}[!t]
   \centering
   \includegraphics[width=2.5in]{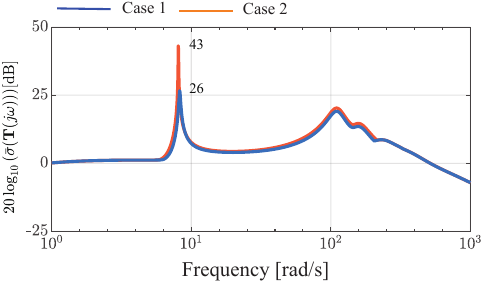}
    \caption{Peak magnitude of the complementary sensitivity function $\mathbf{T}(j\omega)$ under stable operating conditions (Case~1 and Case~2).}
    \label{fig:msv_peak}
\end{figure}

The Generalized Nyquist criterion provides an independent quantitative confirmation of the stability degradation \cite{AGM}. As shown in Fig.~\ref{fig:nyquist}(a), the shortest distance from the eigenvalue loci of $\mathbf{L}(j\omega)=\mathbf{J}_{\mathrm{NET}}(j\omega)\mathbf{K}_{\mathrm{VSC}}(j\omega)$ to the critical point $(-1, j0)$ is 0.25 in Case~1. Consistent with the $M_T$ escalation and the leftward shift of the NMP zero, this margin shrinks to 0.074 in Case~2 (Fig.~\ref{fig:nyquist}(b)). In Case~3, the eigenvalue loci enclose the critical point (Fig.~\ref{fig:nyquist}(c)), confirming the system instability. The consistent correspondence among the $z_0$ location, the $M_T$ peak, and the Nyquist margin establishes the NMP zero as a reliable scalar proxy for the multivariable stability of the multi-converter system. The results show that the migration of $z_0$ toward the origin directly causes the reduction in the stability margin. Therefore, unlike the closed-loop $M_T$ and Nyquist metrics, this NMP zero proxy originates exclusively from the network Jacobian and is independent of individual converter parameters. It allows the stability margin to be evaluated directly through a closed-form expression constructed from grid admittance and nodal power injections, effectively bypassing the black-box barrier of commercial equipment.

\begin{figure*}[!t]
   \centering
   \includegraphics[width=6.1in]{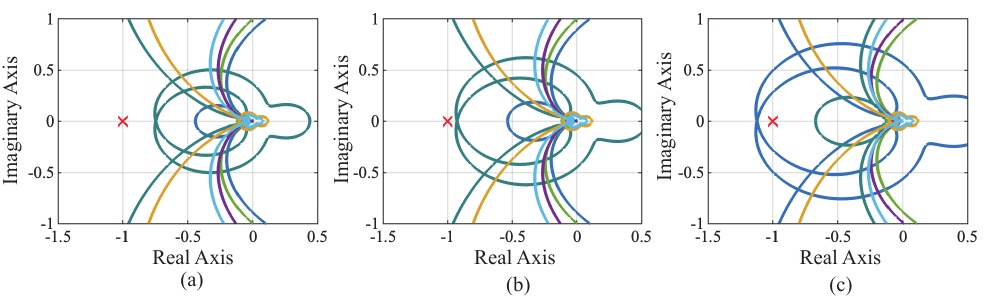}
    \caption{Generalized Nyquist curves of the loop gain matrix $\mathbf{L}(j\omega)$ under varying operating conditions: (a) Case~1, (b) Case~2, and (c) Case~3.}
    \label{fig:nyquist}
\end{figure*}

Time-domain transient responses corroborate these frequency-domain predictions. Following a 0.02~p.u. external grid voltage dip, Fig.~\ref{fig:power_base1}(a) shows that all three VSCs exhibit well-damped recovery in Case~1, with the active power oscillation settling within approximately 0.3~s. In Case~2, Fig.~\ref{fig:power_base1}(b) reveals sustained oscillations with reduced damping, consistent with the near-marginal Nyquist margin of 0.074. In Case~3, Fig.~\ref{fig:power_base1}(c) confirms that the active power diverges, verifying the instability predicted by both the NMP zero analysis and the Generalized Nyquist encirclement.

\begin{figure}[!t]
   \centering
   \includegraphics[width=2.5in]{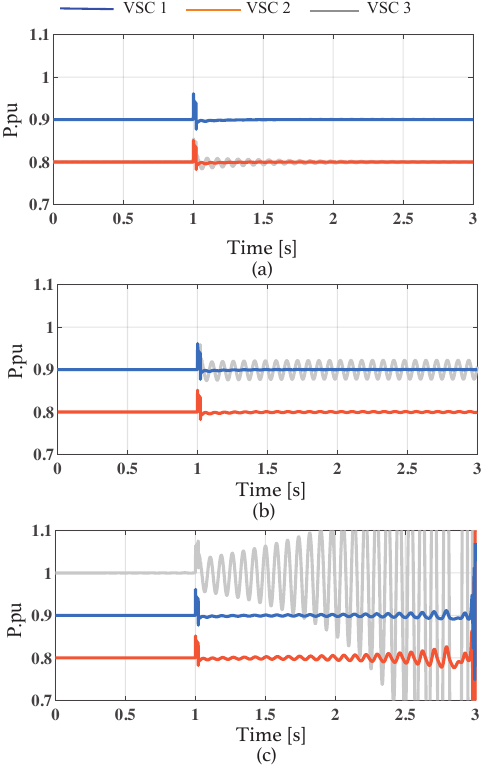}
    \caption{Time-domain active power responses of the VSCs following a grid voltage dip: (a) Case~1, (b) Case~2, and (c) Case~3.}
    \label{fig:power_base1}
\end{figure}

\subsection{Effectiveness of the Zero Reshaping Strategy}

Having established that the NMP zero location governs the stability margin, this subsection validates the zero reshaping strategy proposed in Section~\ref{sec:reshaping}. In the unstable Case~3 ($P_3=1.0$), the critical NMP zero is located at $z_0 = 341$~rad/s. To restore stability, the reactive power outer loop at a selected node is replaced by voltage droop control $\Delta Q_i = -k_i \Delta U_i$ with $k_i=10$. The core prediction of the reshaping strategy is that deploying the droop at the node with the largest real participation factor $\ReOp(p_i)$ will produce the maximum rightward shift of $z_0$ and, consequently, the greatest improvement in stability margin.

Evaluating the participation factors from the eigenvectors of $\mathbf{J}_{\mathrm{NET}}(z_0)$ at $z_0 = 341$ reveals a non-uniform structure, with $\ReOp(p_1) = 0.0005$, $\ReOp(p_2) = 0.0019$, and $\ReOp(p_3) = 0.497$. Node~3 dominates by two orders of magnitude, and the reshaping strategy therefore identifies it as the optimal control location.

\begin{table}[!t]
\renewcommand{\arraystretch}{1.3}
\caption{Steady-State Profiles and NMP Zeros with Droop Control Applied at Different Nodes}
\label{tab:droop_conditions}
\centering
\begin{tabular}{lcccc}
\toprule
Node & $U_1$ (p.u.) & $U_2$ (p.u.) & $U_3$ (p.u.) & NMP Zero \\
\midrule
1 & $1.00 \angle 8.66^\circ$ & $1.05 \angle 6.94^\circ$ & $0.87 \angle 35.41^\circ$ & 356 \\
2 & $1.05 \angle 8.31^\circ$ & $1.00 \angle 7.19^\circ$ & $0.89 \angle 33.70^\circ$ & 387 \\
3 & $1.06 \angle 8.19^\circ$ & $1.06 \angle 6.81^\circ$ & $1.00 \angle 31.18^\circ$ & 863 \\
\bottomrule
\end{tabular}
\end{table}

To rigorously validate this prediction, identical voltage droop control is separately deployed at each of the three nodes. Table~\ref{tab:droop_conditions} summarizes the resulting steady-state profiles and NMP zeros. The experimental outcome exhibits a clear positive correlation between the participation factor magnitude and the achieved zero shift. Node~1 ($\ReOp(p_1) = 0.0005$) yields a marginal shift from 341 to 356~rad/s; Node~2 ($\ReOp(p_2) = 0.0019$) produces a slightly larger shift to 387~rad/s; and Node~3 ($\ReOp(p_3) = 0.497$) achieves a shift to 863~rad/s, a 153\% increase that pushes the critical zero beyond the bandwidth ceiling. This monotonic relationship between control effectiveness and participation factor magnitude confirms the analytical proportionality $\mathrm{d}z_0/\mathrm{d}k_i \propto \ReOp(p_i)$ derived in Section~\ref{sec:reshaping}, and demonstrates that the proposed reshaping method correctly identifies the optimal control location without iterative search or full eigenvalue recomputation.

\begin{figure}[!t]
\centering
\includegraphics[width=2.5in]{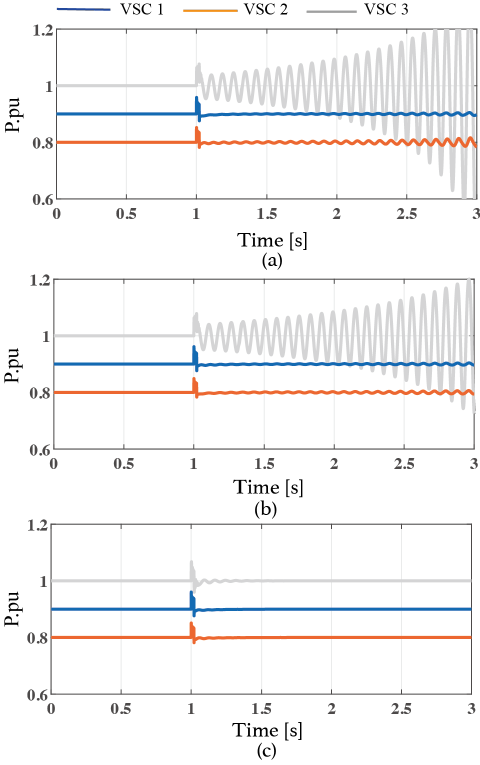}
\caption{Time-domain active power responses with droop control deployed at different nodes: (a) Node~1, (b) Node~2, and (c) Node~3.}
\label{fig:droop_comp}
\end{figure}

The time-domain responses in Fig.~\ref{fig:droop_comp} provide further validation. Following the same 0.02~p.u. voltage dip applied in the baseline analysis, deploying voltage droop control at Node~3 suppresses the oscillations and restores a well-damped response (Fig.~\ref{fig:droop_comp}(c)), consistent with the enlarged NMP zero at 863~rad/s. In contrast, Fig.~\ref{fig:droop_comp}(a) and Fig.~\ref{fig:droop_comp}(b) show that identical control at Nodes~1 and~2 fails to stabilize the system, with persistent divergent oscillations that reflect the insufficient rightward zero shift at these low-participation nodes. These results confirm that the zero reshaping strategy identifies the optimal control location and restores stability by shifting the critical NMP zero away from the origin, while identical control at low-participation nodes fails to achieve this outcome.

\section{Conclusion}
\label{sec:conclusion}

This paper presented a Jacobian-based framework for quantifying and reshaping the NMP zeros that limit stability margins in multi-converter systems. Three results were established. First, the NMP zero locations are determined exclusively by the grid admittance matrix and the steady-state nodal power injections, independent of converter controller parameters. This bypasses the need for proprietary device models and enables stability assessment from network-level data alone. Second, an exponential lower bound on the peak of the complementary sensitivity function was derived as a function of the dominant NMP zero. As this zero approaches the origin, the sensitivity peak increases, imposing a ceiling on achievable stability margins that is independent of controller design. Third, a zero reshaping strategy that requires no iterative search was developed. Specifically, the sensitivity of the dominant zero to a local voltage droop gain is proportional to the real participation factor at each node, which reduces control placement to a scalar ranking across candidate buses.

Simulations on a modified IEEE 9-bus system support these results. As the dominant NMP zero is driven toward the origin, the peak of the complementary sensitivity function increases in agreement with the derived bound, confirming the predicted degradation of stability margins. Deploying voltage droop control at the node with the highest participation factor repositions the zero and restores stable operation, whereas identical control at nodes with low participation factors remains ineffective. These results provide system operators with a toolset, driven by data from the network, for assessing and mitigating stability limitations in grids dominated by converters without requiring device-level models.

\appendices

\section{Derivation of the Network Jacobian Matrix}
\label{app:jacobian}

This appendix derives the dynamic network Jacobian matrix $\mathbf{J}_{\mathrm{NET}}(s)$ based on the analytical approach in \cite{Zyang2020}. We define the steady-state complex nodal voltage as $\mathcal{U}_i = U_{i}e^{j\theta_{i}}$ and partition the matrix into diagonal and off-diagonal blocks.

\paragraph{Diagonal Block of $\mathbf{J}_{\mathrm{NET}}(s)$}
The diagonal block $\mathbf{J}_{ii}(s)$ captures the local power-to-voltage sensitivities. Linearizing the power flow equations produces a constant steady-state power matrix and a summation of dynamic admittance terms. Since the self-coupling phase difference is zero, the coefficient simplifies to a real-valued expression
\begin{equation}
\label{eq:jac_diag_final}
\mathbf{J}_{ii}(s) =
\begin{bmatrix}
-Q_{i} & P_{i}\\
P_{i} & Q_{i}
\end{bmatrix}
+
\sum_{\substack{j=1 \\ j \neq i}}^{N}
\ReOp( \mathcal{U}_i B_{ij} \overline{\mathcal{U}}_j )
\begin{bmatrix}
\alpha(s) & \beta(s) \\
-\beta(s) & \alpha(s)
\end{bmatrix}
\end{equation}
where $\mathcal{U}_i B_{ij} \overline{\mathcal{U}}_j = B_{ij} U_{i} U_{j} e^{j(\theta_{i}-\theta_{j})}$.

\paragraph{Off-Diagonal Block of $\mathbf{J}_{\mathrm{NET}}(s)$}
The off-diagonal block $\mathbf{J}_{ij}(s)$ describes the cross-coupling effects from adjacent node $j$ to node $i$. We derive this block by separating the real and imaginary parts of the effective complex admittance
\begin{equation}
\label{eq:jac_offdiag_final}
\begin{aligned}
\mathbf{J}_{ij}(s) &=
\ReOp( \mathcal{U}_i B_{ij} \overline{\mathcal{U}}_j )
\begin{bmatrix}
\alpha(s) & \beta(s) \\
-\beta(s) & \alpha(s)
\end{bmatrix}
\\
&\quad +
\ImOp( \mathcal{U}_i B_{ij} \overline{\mathcal{U}}_j )
\begin{bmatrix}
\beta(s) & -\alpha(s) \\
\alpha(s) & \beta(s)
\end{bmatrix}.
\end{aligned}
\end{equation}

This formulation demonstrates that the spatial distribution of the Jacobian coefficients corresponds exactly to the elements of the weighted admittance matrix $\mathbf{Y} = \mathbf{U} \mathbf{B}_r \bar{\mathbf{U}}$ defined in the main text. This structure encodes both the network topology and the steady-state voltage profile.

\section{The Relationship Between NMP Zero Reshaping and Participation Factors}
\label{appendix:proof}

This appendix proves that, for a multi-converter system under PQ outer-loop control, $\ReOp(\mathrm{d}z_0/\mathrm{d}k_i) \propto \ReOp(p_i)$.

From \eqref{eq:sensitivity_breakdown}, let $\mathbf{l}$ and $\mathbf{r}$ denote the left and right eigenvectors of $\mathbf{J}_{\mathrm{sys}}$ at the zero $z_0$. For a single-node perturbation at node $i$, the numerator reduces to $p_i = l_i r_i$. Defining $\mathcal{S}_{\text{sys}} = - (\mathbf{l}^{\mathrm{H}} \frac{\partial \mathbf{J}_{\mathrm{sys}}}{\partial s} \mathbf{r})^{-1}$ yields
\begin{equation} \label{eq:local_sens}
    \frac{\mathrm{d} z_0}{\mathrm{d} k_i} = p_i \cdot \mathcal{S}_{\text{sys}}.
\end{equation}

The factor $\mathcal{S}_{\text{sys}}$ is independent of the actuator location $i$. To determine its sign, consider a uniform gain $k$ applied at all nodes ($\mathbf{K} = k\mathbf{I}$). With the normalization $\mathbf{l}^{\mathrm{H}}\mathbf{r}=1$, summing \eqref{eq:local_sens} over all nodes gives
\begin{equation}
    \left( \frac{\mathrm{d} z_0}{\mathrm{d} k} \right)_{\text{uniform}} = \sum_{i=1}^{N} p_i \cdot \mathcal{S}_{\text{sys}} = \mathcal{S}_{\text{sys}}.
\end{equation}

This uniform feedback is equivalent to the admittance augmentation $\mathbf{Y}(s) \to \mathbf{Y}(s) + k\mathbf{I}$, so $z_0$ satisfies the characteristic equation
\begin{equation} \label{eq:char_eq_uniform}
    z_0^2 + 1 = \lambda_j \left( \mathbf{M}(k) \right)
\end{equation}
where $\mathbf{M}(k) = \mathbf{S}^{-1} (\mathbf{Y} + k\mathbf{I}) \overline{\mathbf{S}}^{-1} (\overline{\mathbf{Y}} + k\mathbf{I})$. Implicit differentiation at $k=0$ gives
\begin{equation}
    2 z_0 \frac{\mathrm{d} z_0}{\mathrm{d} k} = \frac{\mathrm{d} \lambda_j}{\mathrm{d} k}.
\end{equation}

Let $\mathbf{u}$ and $\mathbf{v}$ be the left and right eigenvectors of $\mathbf{M}(0)$. Evaluating $\left. \frac{\mathrm{d} \mathbf{M}}{\mathrm{d} k} \right|_{k=0} = 2\mathbf{S}^{-2} \ReOp(\mathbf{Y})$ yields
\begin{equation}
    \frac{\mathrm{d} \lambda_j}{\mathrm{d} k} = \frac{\mathbf{u}^{\mathrm{H}} \left( 2\mathbf{S}^{-2} \ReOp(\mathbf{Y}) \right) \mathbf{v}}{\mathbf{u}^{\mathrm{H}}\mathbf{v}}.
\end{equation}

Since $\ReOp(\mathbf{Y})$ is the network conductance matrix and is positive definite for any passive network, $\ReOp(\mathrm{d} \lambda_j / \mathrm{d} k) > 0$. Because $z_0 > 0$, it follows that $\ReOp(\mathcal{S}_{\text{sys}}) > 0$. Returning to \eqref{eq:local_sens},
\begin{equation}
    \ReOp\left( \frac{\mathrm{d} z_0}{\mathrm{d} k_i} \right) \propto \ReOp(p_i).
\end{equation}

\section{System Parameters}
\label{app:parameters}

\begin{table}[htbp]
\renewcommand{\arraystretch}{1.5} %
\centering
\caption{NETWORK LINE PARAMETERS OF THE MODIFIED IEEE 9-BUS SYSTEM (P.U.)}
\label{tab:network_parameters}
\begin{tabular}{ccc}
\toprule
$X_{15}=0.0411$ & $X_{45}=0.0086$ & $X_{78}=0.01065$ \\
\midrule
$X_{26}=0.0250$ & $X_{46}=0.0181$ & $X_{69}=0.00665$ \\
\midrule
$X_{37}=0.1510$ & $X_{58}=0.01065$ & $X_{79}=0.00665$ \\
\bottomrule
\end{tabular}
\end{table}

\end{document}